\documentclass{article}

\usepackage[utf8]{inputenc}%(only for the pdftex engine)
\usepackage[margin=2.5cm]{geometry}
\usepackage{graphicx} % Required for inserting images
\usepackage{xcolor}
\usepackage{amsmath}
\usepackage{amssymb}
\usepackage{tikz}
\usetikzlibrary{arrows.meta}
\usepackage{bm}
\usepackage{bbm}
\usepackage{algorithm}
\usepackage{algpseudocode}
\newcommand{\indep}{\perp\!\!\!\perp}
\usepackage{caption}
\usepackage{subcaption}
\usepackage{amsthm}
\usepackage{natbib}

\newtheorem{proposition}{Proposition}
\newtheorem{corollary}{Corollary}

\newtheorem{definition}{Definition}
\newtheorem*{definition*}{Definition}
\usepackage{tikz-cd}

\newcommand{\DF}[1]{P_{#1}}

\author{Gustav Jonzon \and Erin E Gabriel \and Arvid Sjölander \and Michael C Sachs \\ 
{\small
  Department of Medical Epidemiology and Biostatistics, Karolinska Institutet, Sweden} \\
  {\small Section of Biostatistics, Department of Public Health, University of Copenhagen, Denmark
  }}
  \title{Adding covariates to bounds: What is the question?}

\begin{document}

\maketitle

\begin{abstract}
{Symbolic nonparametric bounds for partial identification of causal effects now have a long history in the causal literature. Sharp bounds, bounds that use all available information to make the range of values as narrow as possible, are often the goal. For this reason, many publications have focused on deriving sharp bounds, but the concept of sharp bounds is nuanced and can be misleading. In settings with ancillary covariates, the situation becomes more complex. We provide clear definitions for pointwise and uniform sharpness of covariate-conditional bounds, that we then use to prove some general and some specific to the IV setting results about the relationship between these two concepts. As we demonstrate, general conditions are much more difficult to determine and thus, we urge authors to be clear when including ancillary covariates in bounds via conditioning about the setting of interest and the assumptions made.}
\end{abstract}
 % \keywords{partial identification, causal bounds, instrumental variables}
 % \classification{MSC 2020: 62D20}

 % \communicated{...}
 % \dedication{...}

 % \journalname{Journal of Causal Inference}
%\DOI{DOI}
%  \startpage{1}
%  \received{..}
%  \revised{..}
%  \accepted{..}

%  \journalyear{2025}
%  \journalvolume{}
%  \journalissue{1}

\section{Introduction}

In many fields of applied research, the aim of a study is to estimate the effect of some exposure on some outcome. However, beyond the setting of a randomized controlled trial (RCT) with perfect compliance, (point-)identification of the target estimand requires additional strong and untestable assumptions. If one cannot convincingly justify such assumptions, partial identification offers an alternative.

In settings with discrete variables, there have been many partial identification bounds derived in the literature \citep{manski1990nonparametric, doi:10.1080/01621459.1997.10474074, doi:10.1080/01621459.2018.1434530, gabriel2022causal}. However, not all of these bounds are the narrowest possible given all available information, referred to as \textit{sharp} or \textit{tight} \citep{WOS:000801008700001}. The concept of sharpness is relative to the assumptions one is willing to make, as was clearly outlined in \citet{doi:10.1080/01621459.2018.1434530} for the binary instrumental variable (IV) setting. %Given a set of assumptions, sharp bounds are obtained by exploiting any and all information derivable from the assumptions and observations. Partially exploiting further information via additional or stronger assumptions can narrow the bounds, but the narrower bounds may or may not be sharp, given the additional information. When the observables include covariates, they provide additional information, and based on the assumptions one is willing to make about the additional variables, bounds derived in the same setting without covariates, but now conditioned on them, may or may not be sharp. This is a much less well-studied problem than deriving bounds in settings without covariates, but an important one, since multiple works now have suggested using bounds as loss functions for optimizing individual treatment regimens in observation settings \citep{laurendeau2024, pu2021estimating}. 

If the covariates are discrete and the restrictions outlined in \citet{WOS:000801008700001} are satisfied, then the methods therein yield sharp bounds given the assumptions encoded in the independent-error nonparametric structural equation model (NPSEM) corresponding to a DAG containing the covariates. Similarly, structural assumptions were used to derive the well-known and extensively used Balke-Pearl bounds for the binary IV setting. However, this procedure can be computationally demanding and does not currently extend to continuous covariates. 

An alternative method was suggested in \citet{https://doi.org/10.1002/sim.2766} for the binary IV setting, where bounds on the estimand were obtained by averaging the stratum-specific nautral IV bounds over strata of a covariate. This procedure has the appealing feature of being easily extended to continuous covariates, although estimation in such settings may require further assumptions. When the covariate is a parent of the exposure and outcome and independent of the IV in the binary IV setting, \citet{https://doi.org/10.1002/sim.2766} show that the covariate-averaged bounds are valid and no wider than the corresponding bounds that disregard the covariate. They also commented on some conditions on the observed distributions for which the resulting bounds will be narrower. However, they did not discuss sharpness. As the additional restrictions on the causal model implied by the assumptions on how the covariates fit in the DAG are not used in deriving the bounds, it is challenging to determine if the bounds are sharp given this new information. We show in this paper that the covariate-averaged Balke-Pearl bounds in this setting are in fact sharp in a specific setting with all binary observed variables.

A recent paper by \citet{levis2023covariate} proves that even beyond the setting described by \citet{https://doi.org/10.1002/sim.2766}, the covariate-averaged Balke-Pearl bounds are sharp when making two assumptions in the IV setting with covariates. The sharpness of the covariate-averaged Balke-Pearl bounds holds specifically under these assumptions and \emph{no additional assumptions}. As we will illustrate, this setting differs from the scenario where all variables are assumed to have some known connectivity within a causal graph, for example. %for which these assumptions and potentially others hold. 
As there is more information in the latter causal model setting, one must consider sharpness concerning that information making it more difficult to account for all of it in the formulation of the bounds. Thus, making fewer assumptions may result in wider, but still sharp bounds. 

In this article we specify the notions with the precision needed to formally show the relationship between sharpness conditional on covariates, marginal sharpness, and when covariate-averaging does provide sharp bounds even in settings where additional assumptions are made about the covariate. As we will show with examples, once one is willing to assume something about how the additional covariates fit in a DAG, the sharp bounds derived using that information will often be narrower than those obtained using covariate averaging, and they will be the same only in certain settings. This highlights that the notion of sharpness is quite elusive and riddled with potential pitfalls. In lieu of general graphical or distributional criteria for sharpness, which seem difficult to formulate even in the simplest settings, rigorous investigation within the specific setting is needed when conditioning on or averaging over covariates that were not used in the derivation of the bounds. Deriving general conditions for bounds to be sharp after manipulations such as conditioning and covariate-averaging under the largest set of plausible assumptions is a much more difficult task. It is a worthwhile effort, however, since reliable causal bounds tend to be wide, but exhaustively exploiting the assumptions may make them more informative.

\section{Preliminaries and Notation}

Let $Y$ be an outcome and $X$ be an exposure. Using the Neyman-Rubin counterfactual notation, let $Y(X = x)$ denote the value of $Y$ had, potentially counter to fact, the exposure $X$ been set to $x$. We will assume consistency and the lack of interference, such that for a given subject $i$, $(Y_i|X_i=x) = Y_i(x)$, and there is one counterfactual per subject per intervention, independent of the treatment or outcome of other subjects. 

Let $G$ be a causal model that may consist of NPSEM \citep{Pearl2000-PEAC-2} or simply a set of assumptions about the variables and setting of interest that also defines the observable data structure. Let $\boldsymbol{V} = \boldsymbol{O} \cup \boldsymbol{U}$ where $\boldsymbol{U}$ are unobserved and $\boldsymbol{O}$ are observed, and $S$ is an element of $\boldsymbol{O}$ taking values in $\mathcal{S}$. Let $\mathcal{P}$ be the set of joint distributions on $\boldsymbol{V}$ that are compatible with $G$. For a given $\DF{\boldsymbol{V}} \in\mathcal{P}$, let $\DF{\boldsymbol{U}}$ be the marginal distribution of $\boldsymbol{U}$, $o(\DF{\boldsymbol{V}})=\int \DF{\boldsymbol{V}} \, d\DF{\boldsymbol{U}}$ the mapping $o:\mathcal{P}\to\mathcal{O}$, $\mathcal{O} = \{o(\DF{\boldsymbol{V}}): \DF{\boldsymbol{V}} \in \mathcal{P}\}$, and $\DF{\boldsymbol{O}} \in \mathcal{O}$ be the marginal distribution of $\boldsymbol{O}$. For a given $\DF{\boldsymbol{O}}\in\mathcal{O}$, let $o^{-1}(\DF{\boldsymbol{O}}) = \{\DF{\boldsymbol{V}} \in \mathcal{P}: \int \DF{\boldsymbol{V}} \, d\DF{\boldsymbol{U}} = \DF{\boldsymbol{O}}\}$ be the set of distributions with the same observed margin (i.e., the inverse image of $\DF{\boldsymbol{O}}$).   

Let for each \(s\in\mathcal{S}\), $\DF{\boldsymbol{V}|s}$ be the distribution of $\boldsymbol{V}$ given $S = s$, and likewise $\DF{\boldsymbol{U}|s}$ the distribution of $\boldsymbol{U}$ given $S = s$. An observed conditional distribution is denoted $\DF{\boldsymbol{O}|s} = o_s(\DF{\boldsymbol{V}})$ which is an element of $\mathcal{O}_s = \{\int \DF{\boldsymbol{V}|s}\, d\DF{\boldsymbol{U}|s}: \DF{\boldsymbol{V}} \in \mathcal{P}\}$. The inverse image of the conditional observed distribution is $o_s^{-1}(\DF{\boldsymbol{O}|s}) = \{\DF{\boldsymbol{V}} \in \mathcal{P}: \int \DF{\boldsymbol{V}|s}\, d\DF{\boldsymbol{U}|s} = \DF{\boldsymbol{O}|s}\}$. Let \(\DF{S}\) be the marginal distribution of $S$.

Let $\theta: \mathcal{P} \rightarrow \mathbb{R}$ be the causal parameter of interest, which is assumed to be a well-defined causal quantity in the causal model. For example, 
\[\theta(\DF{\boldsymbol{V}})=\mathbb{E}[Y(X{=}1)-Y(X{=}0)]\]
and let $L: \mathcal{O} \rightarrow \mathbb{R}$ be a lower bound for $\theta$. For simplicity, we will focus on lower bounds. The definitions and deductions regarding upper bounds are analogous. Let $\theta_{s}: \mathcal{P} \rightarrow \mathbb{R}$ and $L_{s}: \mathcal{O}_s \rightarrow \mathbb{R}$ be the conditional parameter and covariate-conditional lower bound, respectively. The covariate-conditional estimand is a well-defined causal quantity. For example, 
\[\theta_{s}(\DF{\boldsymbol{V}})=\mathbb{E}_{\DF{\boldsymbol{V}}}[Y(X{=}1)-Y(X{=}0)|S=s]\] and note that \(\theta=\int \theta_{s} \, d\DF{S}(s)\). Let $\bar{L}(\DF{\boldsymbol{O}})= \int L_s(\DF{\boldsymbol{O}|s}) \, d\DF{S}(s)$ be the covariate-averaged lower bound of \(\theta\).

A summary of the mappings involved is shown in Figure \ref{fig:Mappings2}.

\begin{figure}[hbt!]
    \centering
    \begin{tikzcd}%[width=0.5\linewidth]{}
        &\rotatebox[origin=r]{-90}{\text{Underlying distributions:}}&\rotatebox[origin=r]{-90}{\text{Observation mapping:}}&\rotatebox[origin=r]{-90}{\text{Observations:}}&\rotatebox[origin=r]{-90}{\text{Lower bound functions:}}&\\\\\\
        \rotatebox[origin=r]{0}{\text{Covariate inclusive model:}}&\mathcal{P}\arrow[rrrr,bend left=50,magenta,"\theta"]\arrow[rr,blue,"o"]\arrow[rrdd,blue,"o_s"]\arrow[rrrr,bend right=70,looseness=1.7,swap,magenta,"\theta_s"]&&\mathcal{O}\arrow[rr,bend left,teal,"L"]\arrow[rr,bend right,teal,"\bar{L}"]\arrow[dd]&&\mathbb{R}\\\rotatebox[origin=r]{0}{\text{Covariate conditioning:}}\\
        \rotatebox[origin=r]{0}{\text{Covariate conditional model:}}&&&\mathcal{O}_s\arrow[rruu,teal,swap,"L_s"]&&
    \end{tikzcd}
    \caption{Spaces and mappings of the models. Observation mappings are in {\color{blue}blue}. Lower bound functions are in {\color{teal}teal}. Estimand functions are in {\color{magenta}magenta}.}
    \label{fig:Mappings2}
\end{figure}
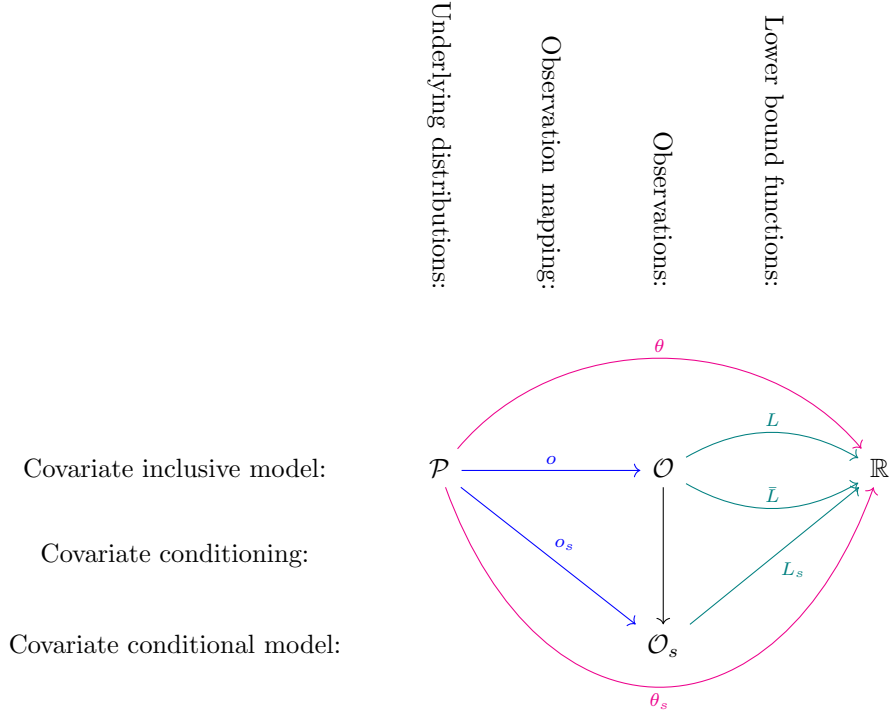

\section{Definitions}
\begin{definition}
    $L$ is a \textbf{valid lower bound} of $\theta$ under $G$, if  for all \(\DF{\boldsymbol{V}}\in\mathcal{P}\), \[L\{o(\DF{\boldsymbol{V}})\}\le\theta(\DF{\boldsymbol{V}}).\]
\end{definition}

\begin{definition}
A valid lower bound $L$ of $\theta$ is \textbf{sharp} under $G$ if for all $\DF{\boldsymbol{O}} \in \mathcal{O}$ and $\varepsilon > 0$ there exists a $\DF{\boldsymbol{V}} \in o^{-1}(\DF{\boldsymbol{O}})$ such that 
$$
\theta(\DF{\boldsymbol{V}}) < L(\DF{\boldsymbol{O}}) + \varepsilon.
$$ 
\end{definition}

Note that sharpness is a property that depends on the given context, namely the assumed causal model $G$ and the form of the observation space $\mathcal{O}$. 
%Recall that $\mathcal{O}$ is the set of joint distributions on the observable variables, i.e., an element of $\mathcal{O}$ is of the form $\{P(\boldsymbol{O} = \boldsymbol{o}):  \boldsymbol{o} \mbox{ a value of }\mathbf{O}\}$, where $P$ may be either a probability mass function or a probability density function. 
%Similarly, an element of $\mathcal{O}_s$ is of the form $\{P(\boldsymbol{O}\smallsetminus S = \boldsymbol{o}\smallsetminus s \vert S = s):\boldsymbol{o} \mbox{ a value of }\mathbf{O}\}.$ 
%Other situations could be considered, 
Observation mappings other than $o$, which returns the full joint distribution, could be considered. For example, \citet{ramsahai2007causal} considered the ``bivariate'' observation situation in the IV setting, where the observed distributions take the form $\{P(X = x, Z = z), P(Y = y, Z = z): \mbox{ for all } x,y,z \in \{0,1\}\}$ rather than the full joint distribution of $X, Y, Z$. Although \textit{tight} and \textit{sharp} have been used interchangeably in the literature, we will use \textit{sharp} moving forward for consistency of language. 
%Note that sharpness fits the common challenge-response pattern in analysis: for a given challenge \(\varepsilon>0\), the response (generally depending on the challenge) is some \(N\in\mathbb{N}\) for convergence, some \(\delta>0\) for continuity and some joint distribution \(\DF{\boldsymbol{V}}\in\mathcal{P}\) for sharpness. 

Now that we have a definition of sharp, let $L^{co}$ $(\DF{\boldsymbol{O}})$ be the sharp lower bound for $\theta$ in $G$, where we use $co$ to indicate covariate-optimal, since $G$ contains $S$. We now consider conditional versions of valid and sharp.

\begin{definition}
    The set \(\{L_s:s\in\mathcal{S}\}\) is 
    \textbf{pointwise valid} for the set \(\{\theta_s: s \in S\}\) under $G$, if for all \(\DF{\boldsymbol{V}}\in\mathcal{P}\) and $P_S$-almost every $s\in\mathcal{S}$, \[L_s\{o_s(\DF{\boldsymbol{V}})\}\le\theta_s(\DF{\boldsymbol{V}}).\]
\end{definition}

\begin{definition}
     A pointwise valid set \(\{L_s:s\in\mathcal{S}\}\) for the set \(\{\theta_s: s \in S\}\) under $G$ is 
    \textbf{pointwise sharp} for the set \(\{\theta_s: s \in S\}\) under $G$, if for all $\DF{\boldsymbol{O}|s} \in \mathcal{O}_s$ and $P_S$-almost every $s\in\mathcal{S}$ and $\varepsilon_s > 0$, there exists a $\DF{\boldsymbol{V}} \in o_s^{-1}(\DF{\boldsymbol{O}|s})$  \[\theta_s(\DF{\boldsymbol{V}}) < L_s(\DF{\boldsymbol{O}|s}) + \varepsilon_s.\] 
\end{definition}

\begin{definition}
    A pointwise valid set \(\{L_s:s\in\mathcal{S}\}\) for the set \(\{\theta_s: s \in S\}\) under $G$ is \textbf{uniformly sharp} for the set \(\{\theta_s: s \in S\}\) under $G$ if for all $\DF{\boldsymbol{O}} \in \mathcal{O}$ and $\varepsilon > 0$ there exists a $\DF{\boldsymbol{V}} \in o^{-1}(\DF{\boldsymbol{O}})$ such that for $\DF{S}$-almost every $s \in \mathcal{S}$,
    \[\theta_s(\DF{\boldsymbol{V}}) <  L_s(\DF{\boldsymbol{O}|s}) +\varepsilon,
    \] 
where $\DF{\boldsymbol{O}|s} = o_s(\DF{\boldsymbol{V}})$ for all $s$.
\end{definition}

Note that uniform sharpness requires that there exists a single distribution $\DF{\boldsymbol{V}}$ that fulfills the pointwise sharpness requirements of all levels of $S$ simultaneously. Figure \ref{fig:diagram} illustrates the distinction. Uniform sharpness is a strictly stronger requirement than pointwise sharpness and implies pointwise sharpness. Pointwise validity and sharpness are local properties that apply to each \(s\in\mathcal{S}\) separately. In contrast, uniform sharpness is a global property that pertains to the whole set \(\mathcal{S}\) and the causal model $G$ containing the covariate $S$.

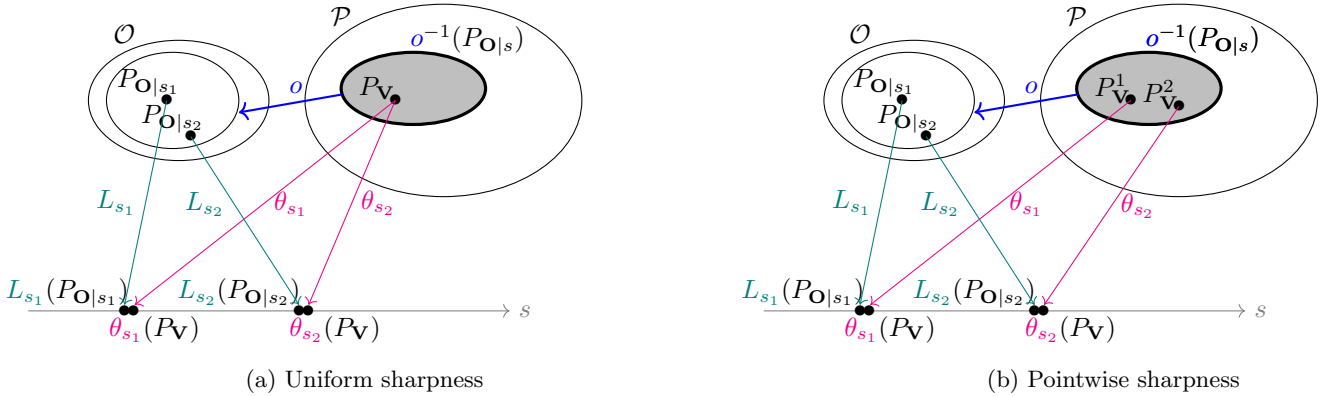
\begin{figure}[ht]
    \centering
    \begin{subfigure}[b]{0.4\textwidth}
        \centering
        \hspace*{-60pt}
        \begin{tikzpicture}[scale=0.8]
            \draw (1,6) ellipse (1.5 and 1);
            \node at (0.1,7.1) {$\mathcal{O}$};
            \node at (0.8,6.0) {$\bullet$};
            \node at (0.5,6.3) {$P_{\mathbf{O}|s_1}$};
            \node at (1.2,5.4) {$\bullet$};
            \node at (0.9,5.7) {$P_{\mathbf{O}|s_2}$};
            \draw (0.9,6.0) ellipse (1.1 and 0.8);
            \node at (3.7,7.4) {$\mathcal{P}$};
            \draw (5.4,6) ellipse (2.3 and 1.6);
            \draw[very thick,fill=lightgray] (4.9,6.2) ellipse (1.2 and 0.6);
            \node at (4.6,6.0) {$\bullet$};
            \node at (4.3,6.2) {$P_\mathbf{V}$};
            \node at (5.8,7.0) {${\color{blue}o}^{-1}(P_{\mathbf{O}|s})$};
            \draw[->,thick,blue] (3.72,6.1) -- (2.0,5.8) node[pos=0.45,above] {$o$};
            \draw[->, gray] (-1.5,2.5) -- (6.5,2.5) node[right] {$s$};
            %\draw[->, gray] (-1.5,0.5) node[left] {$s=1:$} -- (6.5,0.5) node[right] {$\mathbb{R}$};
            \node at (0.1,2.5) {$\bullet$};
            \node at (-0.85,2.8) {${\color{teal}L_{s_1}}(P_{\mathbf{O}|s_1})$};
            \node at (0.25,2.5) {$\bullet$};
            \node at (0.6,2.2) {${\color{magenta}\theta_{s_1}}(P_\mathbf{V})$};
            \draw[->,teal] (0.8,6.0) -- (0.1,2.6) node[pos=0.5,left] {$L_{s_1}$};
            \draw[->,magenta] (4.6,6.0) -- (0.27,2.6) node[pos=0.4,below] {$\theta_{s_1}$};
            \node at (3.0,2.5) {$\bullet$};
            \node at (2.0,2.8) {${\color{teal}L_{s_2}}(P_{\mathbf{O}|s_2})$};
            \node at (3.15,2.5) {$\bullet$};
            \node at (3.6,2.2) {${\color{magenta}\theta_{s_2}}(P_\mathbf{V})$};
            \draw[->,teal] (1.2,5.4) -- (3.0,2.6) node[pos=0.4,left] {$L_{s_2}$};
            \draw[->,magenta] (4.6,6.0) -- (3.16,2.6) node[pos=0.5,right] {$\theta_{s_2}$};
        \end{tikzpicture}
        \caption{Uniform sharpness}
        \label{fig:uniform}
    \end{subfigure}
    \hfill
    \begin{subfigure}[b]{0.4\textwidth}
        \centering
        \hspace*{-70pt}
        \begin{tikzpicture}[scale=0.8]
            \draw (1,6) ellipse (1.5 and 1);
            \node at (0.1,7.1) {$\mathcal{O}$};
            \node at (0.8,6.0) {$\bullet$};
            \node at (0.5,6.3) {$P_{\mathbf{O}|s_1}$};
            \node at (1.2,5.4) {$\bullet$};
            \node at (0.9,5.7) {$P_{\mathbf{O}|s_2}$};
            \draw (0.9,6.0) ellipse (1.1 and 0.8);
            \node at (3.7,7.4) {$\mathcal{P}$};
            \draw (5.4,6) ellipse (2.3 and 1.6);
            \draw[very thick,fill=lightgray] (4.9,6.2) ellipse (1.2 and 0.6);
            \node at (4.6,6.0) {$\bullet$};
            \node at (4.3,6.2) {$P^1_\mathbf{V}$};
            \node at (5.4,5.9) {$\bullet$};
            \node at (5.1,6.1) {$P^2_\mathbf{V}$};
            \node at (5.8,7.0) {${\color{blue}o}^{-1}(P_{\mathbf{O}|s})$};
            \node at (5.8,7.0) {${\color{blue}o}^{-1}(P_{\mathbf{O}|s})$};
            \draw[->,thick,blue] (3.72,6.1) -- (2.0,5.8) node[pos=0.45,above] {$o$};
            \draw[->, gray] (-1.5,2.5) -- (6.5,2.5) node[right] {$s$};
            %\draw[->, gray] (-1.5,0.5) node[left] {$s=1:$} -- (6.5,0.5) node[right] {$\mathbb{R}$};
            \node at (0.1,2.5) {$\bullet$};
            \node at (-0.85,2.8) {${\color{teal}L_{s_1}}(P_{\mathbf{O}|s_1})$};
            \node at (0.25,2.5) {$\bullet$};
            \node at (0.6,2.2) {${\color{magenta}\theta_{s_1}}(P_\mathbf{V})$};
            \draw[->,teal] (0.8,6.0) -- (0.1,2.6) node[pos=0.5,left] {$L_{s_1}$};
            \draw[->,magenta] (4.6,6.0) -- (0.27,2.6) node[pos=0.4,below] {$\theta_{s_1}$};
            \node at (3.0,2.5) {$\bullet$};
            \node at (2.0,2.8) {${\color{teal}L_{s_2}}(P_{\mathbf{O}|s_2})$};
            \node at (3.15,2.5) {$\bullet$};
            \node at (3.6,2.2) {${\color{magenta}\theta_{s_2}}(P_\mathbf{V})$};
            \draw[->,teal] (1.2,5.4) -- (3.0,2.6) node[pos=0.4,left] {$L_{s_2}$};
            \draw[->,magenta] (5.4,5.9) -- (3.16,2.6) node[pos=0.5,right] {$\theta_{s_2}$};
        \end{tikzpicture}
        \caption{Pointwise sharpness}
        \label{fig:conditional}
    \end{subfigure}    
    \caption{Diagram illustrating the difference between pointwise and uniform sharpness. For uniform sharpness, there needs to exist a single underlying distribution \(P_\mathbf{V}\) which maps each conditional estimand arbitrarily close to the corresponding conditional lower bound. For general pointwise sharpness, this possibly requires distinct underlying distributions for each \(s\in\mathcal{S}\). The color schemes follow those in Figure \ref{fig:Mappings2}.}
    \label{fig:diagram}
\end{figure}

\section{General Results}

We first note that pointwise validity implies valid covariate-averaged bounds.

\begin{proposition}
If the set $\{L_s:s\in\mathcal{S}\}$ is pointwise valid for \(\{\theta_s: s \in S\}\) Under $G$, then $\overline{L}$, the bounds obtained by averaging over $\{L_s:s\in\mathcal{S}\}$, is a valid lower bound for \(\theta\) under $G$.
\end{proposition}

%\noindent \textbf{Proof:} Suppose that $\overline{L}(o(P))>\theta(P)$ for some $P\in\mathcal{P}$, i.e., $E_P[L_S(O_S)]>E_P[\theta_S(P)]$, where for each \(s\in\mathcal{S}\), \(O_s\) is \(o(P)\) conditioned on \(S{=}s\). Thus, for at least some \(s\in\mathcal{S}\), $L_s(O_s)>\theta_s$. 
\begin{proof}
    Let \(\DF{\boldsymbol{V}}\in\mathcal{P}\). Then \(L_s(o_s(\DF{\boldsymbol{V}}))\le\theta_s(\DF{\boldsymbol{V}})\) for all \(s\in\mathcal{S}\), so \(\bar{L}(o(\DF{\boldsymbol{V}}))=\int L_s(\DF{\boldsymbol{O}|s}) \, d\DF{S}(s)\le \int \theta_s(\DF{\boldsymbol{V}})\, d\DF{S}(s) =\theta(\DF{\boldsymbol{V}})\).
\end{proof}

\begin{proposition}
The covariate-averaged lower bound $\overline{L}$ is sharp for $\theta$ under $G$ if and only if the set \(\{L_{s}:s\in\mathcal{S}\}\) is uniformly sharp for the set \(\{\theta_s: s \in S\}\) under $G$.
\end{proposition}
A proof is provided in the supplement.

Although not directly stated, Theorem 1 of \citet{levis2023covariate} uses a similar idea, however, in a very specific setting and without highlighting the differences between pointwise and uniform sharpness. In contrast, Proposition 2 is a more general statement that applies to all settings, not just the binary IV setting.  

Similar statements about the sharpness of the covariate-averaged bounds for specific covariate-conditional bounds defined by graphical properties are more difficult to make. We consider several examples for illustration. 

\subsection{Example}

\begin{figure}[h]

\begin{subfigure}[b]{.2\textwidth}
    \begin{tikzpicture}[scale=1.0]
    \node[draw, circle] at (2, 0) (X) {\(X\)};
    \node[draw, circle] at (4, 0) (Y) {\(Y\)};
    %\node[draw, circle] at (-3, 4) (S) {\(S\)};
   % \node[draw, circle] at (-4, 0) (Z) {\(Z\)};
    \node[draw, circle, dashed] at (3,2) (U) {\(U\)};
    \draw [-{Latex[scale=2.0]}] (X) to (Y);
   
    %\draw [-{Latex[scale=2.0]}] (S) to (X);
    %\draw [-{Latex[scale=2.0]}] (S) to (Y);
    \draw [-{Latex[scale=2.0]}] (U) to (X);
    \draw [-{Latex[scale=2.0]}] (U) to (Y);
   % \draw [-{Latex[scale=2.0]}] (Z) to (X);
    %\draw [-, dashed] (S) to (U);
    \end{tikzpicture}
    \caption{Full confounded DAG}
    \label{fig:DAGcon0}
    \end{subfigure}
    \hspace{0.5cm}
    \begin{subfigure}[b]{.2\textwidth}
    \begin{tikzpicture}[scale=1.0]
    \node[draw, circle] at (2, 0) (X) {\(X\)};
    \node[draw, circle] at (4, 0) (Y) {\(Y\)};
    \node[draw, circle] at (0, 0) (S) {$S$};
    %\node[draw, circle] at (-3, 4) (S) {\(S\)};
   % \node[draw, circle] at (-4, 0) (Z) {\(Z\)};
    \node[draw, circle, dashed] at (3,2) (U) {\(U\)};
    \draw [-{Latex[scale=2.0]}] (X) to (Y);
    \draw [-{Latex[scale=2.0]}] (S) to (X);
    %\draw [-{Latex[scale=2.0]}] (S) to (X);
    %\draw [-{Latex[scale=2.0]}] (S) to (Y);
    \draw [-{Latex[scale=2.0]}] (U) to (X);
    \draw [-{Latex[scale=2.0]}] (U) to (Y);
   % \draw [-{Latex[scale=2.0]}] (Z) to (X);
    %\draw [-, dashed] (S) to (U);
    \end{tikzpicture}
    \caption{$S$ model 1}
    \label{fig:DAGcon}
    \end{subfigure}
    \hspace{2cm}
    \begin{subfigure}[b]{.2\textwidth}
    \begin{tikzpicture}[scale=1.0]
    \node[draw, circle] at (0, 0) (X) {\(X\)};
    \node[draw, circle] at (4, 0) (Y) {\(Y\)};
    \node[draw, circle] at (4, 2) (S) {$S$};
    %\node[draw, circle] at (-3, 4) (S) {\(S\)};
   % \node[draw, circle] at (-4, 0) (Z) {\(Z\)};
    \node[draw, circle, dashed] at (0,2) (U) {\(U\)};
    \draw [-{Latex[scale=2.0]}] (X) to (Y);
    %\draw [-{Latex[scale=2.0]}] (S) to (X);
    %\draw [-{Latex[scale=2.0]}] (S) to (Y);
    \draw [-{Latex[scale=2.0]}] (U) to (X);
    \draw [-{Latex[scale=2.0]}] (U) to (Y);
    \draw [-{Latex[scale=2.0]}] (U) to (S);
    \draw [-{Latex[scale=2.0]}] (S) to (Y);
    \draw [-{Latex[scale=2.0]}] (S) to (X);
   % \draw [-{Latex[scale=2.0]}] (Z) to (X);
    %\draw [-, dashed] (S) to (U);
    \end{tikzpicture}
    \caption{$S$ model 2}
    \label{fig:DAGcon2}
    \end{subfigure}
    \caption{DAG models used for illustration.}
\end{figure}
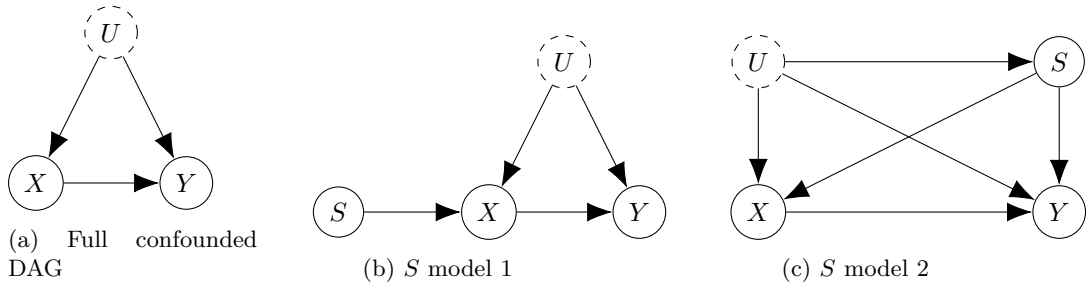
 
    Consider the DAG depicted in Figure \ref{fig:DAGcon0} with \(X\) and \(Y\) binary, with assumptions encoded by NPSEM, given by the DAG. Given an underlying distribution \(\DF{\boldsymbol{V}}\) over \(\boldsymbol{V} = (X,Y,U)\), we observe the distribution \(\DF{\boldsymbol{O}} = o(\DF{\boldsymbol{V}})\) over \(\boldsymbol{O} = (X,Y)\) and \[L(\DF{\boldsymbol{O}}):=-P(X=0,Y=1)-P(X=1,Y=0)\] is a valid lower bound on the population average causal risk difference \(\theta=\mathbb{E}_{\DF{\boldsymbol{V}}}[Y(X{=}1)-Y(X{=}0)]\) that is also well-known to be sharp \citep{robins1989analysis}.
    
Consider now $S$ model 1 and $S$ model 2 depicted in Figures \ref{fig:DAGcon} and \ref{fig:DAGcon2} and the NPSEM they encode, respectively. In both $S$ models 1 and 2, for any value $s$ of $S$, the conditional bound corresponding to $L(\DF{\boldsymbol{O}})$ above, is \(L_s(\DF{\boldsymbol{O}|s})=-P(X=0,Y=1|S=s)-P(X=1,Y=0|S=s)\). Each $L_s(\DF{\boldsymbol{O}|s})$ is sharp, so the set $\{L_s:s \in S\}$ is pointwise sharp under both $S$ models 1 and 2. 

In $S$ model 1, the set $\{L_s:s \in \mathcal{S}\}$ is not uniformly sharp, and the covariate-averaged bound is not sharp for $\theta$. The latter is immediately clear when one considers that if one observes all joint distributions for $S$ model 1, for a binary $S$, the sharp bounds for $\theta$ are the Balke-Pearl bounds. As $\theta = \theta_s$ for all $s$, the lack of uniform sharpness can be seen by positing a $\DF{\boldsymbol{O}}$ under $S$ model 1 such that $\DF{\boldsymbol{O}|s}$ is $P(X=0,Y=1|S=1)=1$ and $P(X=1,Y=1|S=0)=1$, making $L_1=-1$ and $L_0=0$ and thus, the common value of $\theta_0=\theta_1=\theta$ cannot generally attain these bounds at the same time. 

In S model 2, consider a binary $S$, the set $\{L_s:s \in \mathcal{S}\}$ is both pointwise and uniformly sharp. Thus, given \(\DF{\boldsymbol{O}}\in\mathcal{O}\) there exists a \(\DF{\boldsymbol{V}}\in o^{-1}(\DF{\boldsymbol{O}})\) such that \(\bar{L}(\DF{\boldsymbol{O}})=P(X=0,Y=1)-P(X=1,Y=0)\), so the covariate-averaged bounds are sharp for $\theta$. However, \(\bar{L}(\DF{\boldsymbol{O}})=L(\DF{\boldsymbol{O}})\), so although \(\bar{L}(\DF{\boldsymbol{O}})\) is sharp, it is not narrower than $L(\DF{\boldsymbol{O}})$ and thus covariate averaging does not improve the width of the bounds. 

We can show that \(\bar{L}(\DF{\boldsymbol{O}})=L(\DF{\boldsymbol{O}})=L^{co}(\DF{\boldsymbol{O}})\) using \texttt{causaloptim} \citep{WOS:001216733200016}. In these settings the linear programming method used in \texttt{causaloptim} has been proven to be sharp \citep{WOS:000801008700001}. Thus $L^{co}$ in $S$ model 2 for a binary $S$ is given by:
\begin{eqnarray*}
      -P(X=1,Y=0,S=0)-P(X=1,Y=0,S=1)- \\ P(X=0,Y=1,S=0)-P(X=0,Y=1,S=1),
\end{eqnarray*} which by simplification is equal to $L(\DF{\boldsymbol{O}})=\bar{L}(\DF{\boldsymbol{O}})$. Thus, the covariate-averaged, the covariate-optimal, and the covariate-marginal bounds all coincide in this case.

\section{IV setting}

\subsection{Preliminaries and Notation}
 Suppose that \(X\), \(Y\), and $Z$ are measured binary variables, and the unmeasured set of variables $\boldsymbol{U}$ is unrestricted, e.g., continuous or multilevel categorical, or a combination. We will consider $\theta$ and $\theta_s$, to be the causal risk difference and the causal risk difference conditional on $s$, and the set of causal estimands \(\{\theta_s: s \in S\}\) is obtained by conditional on all levels of $S$. There are three commonly cited core conditions for \(Z\) to be a valid IV for \((X,Y)\): \(Z\not\indep X\) (relevance), \(Z\indep Y|\{X,U\}\) (exclusion restriction) and \(Z\indep U\) (latent independence) \citep{didelez2010assumptions}, however, other assumptions have been used in the literature. The IV model has been extensively studied and applied in many contexts.
Notably, \citet{doi:10.1080/01621459.2018.1434530} contains a thorough investigation of the case where $Z$, $X$ and $Y$ are binary, covering and comparing commonly used sets of assumptions in different causal paradigms and deriving bounds under each. We will make the following assumptions.

\noindent\textbf{Setting Assumptions}
    \begin{enumerate}
        \item \(Z\indep \{Y(z), X(z)\}\mid S{=}s\) and 
        \item  $Y(x,z) = Y(x) $ for all \(z\in\mathcal{Z},x\in\mathcal{X}\)
        \end{enumerate}

Assumptions 3 and 4 of \citet{levis2023covariate} are our Assumptions 1 and 2; they additionally invoke consistency and no interference, as we do, but also positivity.

We now specifically define $L_s$ as the Balke-Pearl bounds, i.e. all binary setting IV bounds, where the probabilities in the bounds expressions are all conditional on $S = s$. Then the set of bounds $\{L_{s} : s \in S\}$ is obtained by conditioning on all levels of $S$. Let covariate-average bound $\overline{L}$ be the bounds obtained by averaging over the set $\{L_{s} : s \in S\}$ based on the Balke-Pearl bounds as defined above. Let $\mathcal{G}$ be the set of causal models $\{G\}$ on variables $\{X,S,Y,Z,U\}$ where  Assumptions 1-2, consistency, and no interference hold.

\subsection{Results}
\begin{proposition}
The set \(\{L_{s}: s \in S\}\) based on the Balke-Pearl bounds are pointwise valid for the set $\{\theta_s: s \in \mathcal{S}\}$ under $G$ for any $G\in \mathcal{G}$.
\end{proposition}
A proof is in the Supplementary Materials. 

\begin{corollary}
The covariate-averaged lower bound $\overline{L}$ averaged over \(\{L_{s}: s \in \mathcal{S}\}\) based on the Balke-Pearl bounds is valid for $\theta$ under $G$ for any $G\in\mathcal{G}$.  \end{corollary}

\noindent This immediately follows from Proposition 3 and Proposition 1.\\

This might seem to imply by the same reasoning that the set of lower bounds \(\{L_{s}: s \in \mathcal{S}\}\)  will also be pointwise sharp for the conditional estimands, since the Balke-Pearl bounds are also known to be sharp. However, as outlined above, sharpness of the lower bound is specific to the causal model. Conditioning on $S$ may imply additional constraints on the counterfactuals or observables not included in the Balke-Pearl optimization, if one is willing to assume something about $S$ other than Assumptions 1 and 2, which is generally the case. \\

\noindent\textbf{Theorem 1 of \citet{levis2023covariate}} restated in terms of our pointwise and uniform sharpness concepts states that 
for a causal model ${G}$ that makes their assumptions 1-4 and no other assumptions, the set \(\{L_{s}: s \in \mathcal{S}\}\), based on the Balke-Pearl bounds, will be pointwise and uniformly sharp for $\{\theta_s: s \in S\}$ under ${G}$.

We highlight that we do not state for each $G$ in $\mathcal{G}$ here, as once one assumes anything about $S$ in the causal models, for example, that $S$ is an additional IV as in Figure \ref{fig:DAGcon}, this theorem no longer holds as stated. To see the difference, we go back to the definition of uniformly sharp. 

The causal model defined by a DAG where $S$ has no connections in the DAG, i.e. $S\indep\mathbf{V}\smallsetminus\{S\}$, is included in the set of models $\mathcal{G}$ and thus is part of their causal model $G$ where their assumptions 1-4 hold and nothing else. As for any distribution consistent with this causal model, neither $\theta_s$ nor $L_s$ depend on $s$, and $\theta_s=\theta$, $L_s=L$, where $L$ is the Balke-Pearl bound, which is sharp in this case. By the definition of sharpness, only one unobserved distribution that attains the bound needs to exist for each observed distribution, and thus if there is such a distribution then the bound is sharp under $G$. Although in the assumed setting of \citet{levis2023covariate}, the set of bounds  \(\{L_{s}: s \in \mathcal{S}\}\) will be both pointwise and uniformly sharp under their causal model $G$ only defined by their assumptions 1-4, pointwise sharpness does not in general imply uniform sharpness. We have demonstrated this in general and will exemplify it below in the IV setting.

This does not mean that the covariate-averaging will not improve the width of the bounds, making them narrower while remaining valid in many cases, but they may not be sharp under the causal model. 
%As shown in \citet{https://doi.org/10.1002/sim.2766}, for all binary observed $\{Z,X,Y\}$ and categorical $S$, under Assumptions 1-2, if \(S\indep Z\) then the Balke-Pearl lower bound on \(\theta\) is at most the covariate-averaged lower bound obtained by first conditioning on $S{=}s$ in the Balke-Pearl lower bound and then averaging over the resulting values. 
\citet{https://doi.org/10.1002/sim.2766} showed that in the IV setting, for all binary observed $\{Z,X,Y\}$ and categorical $S$, if \(S\indep Z\), there exist distributions \(\DF{\boldsymbol{V}}\in\mathcal{P}\) for which the covariate-averaged bounds obtained from the covariate-conditional lower and upper bounds are narrower than the Balke-Pearl bounds ignoring $S$, and that they are never wider. We demonstrate this in simulations. However, \citet{https://doi.org/10.1002/sim.2766} did not mention the sharpness of these bounds. 

We show in the following Proposition that in the setting of Figure \ref{fig:simset5}, for a binary $S$, the covariate-averaged Balke-Pearl bounds are sharp. 

We need to define a specific causal model, $G'$ to be the set of NPSEMs encoded by the DAG in Figure \ref{fig:simset5}, with variables $\{X,S,Y,Z,U\}$ where all but $U$ are observed and among the observed all are binary. These are as follows:
\begin{eqnarray*} 
z &=& g_{Z}(\epsilon_{z}) \\ 
s &=& g_{S}(\boldsymbol{u}, \epsilon_{s})\\
x &=& g_{X}(\boldsymbol{u}, z, s, \epsilon_{x}) \\
y &=& g_{Y}(\boldsymbol{u}, x, s, \epsilon_y). 
\end{eqnarray*}

\begin{proposition}
 The set \(\{L_{s}: s \in \mathcal{S}\}\), which are the Balke-Pearl bounds conditioned on $S=s$, are uniformly sharp for the set \(\{\theta_{s}: s \in \mathcal{S}\}\) under $G'$ and the lower covariate-averaged bound $\overline{L}$ is valid and sharp for $\theta$ under $G'$.   
\end{proposition}

Although the complete proof is provided in the Supplementary Materials,  we demonstrate that the covariate-averaged bounds coincide with known covariate-optimal bounds in $G'$.

To simplify notation, we will consider the lower bound on \(\theta(\DF{\boldsymbol{V}}):=\mathbb{E}_P[Y(X{=}0)]\). The upper bound follows an analogous pattern, as do bounds on \(\mathbb{E}_P[Y(X{=}1)]\), and bounds on the causal risk difference can then be obtained from these components.

For compact notation, let for each \(x,y,z\in\{0,1\}\), \(p_{yx.z}:=P(Y{=}y,X{=}x|Z{=}z)\) and for each \(s\in\{0,1\}\), \(p_{yxs.z}:=P(Y{=}y,X{=}x,S{=}s|Z{=}z)\), \(p_{yx.zs}:=P(Y{=}y,X{=}x|Z{=}z,S{=}s)\), \(p_{s.z}:=P(S{=}s|Z{=}z)\) and \(p_{s}:=P(S{=}s)\).

 \(o(\DF{\boldsymbol{V}\setminus S})\) is the distribution of \((X,Y,Z)\) obtained from \(\DF{\boldsymbol{V}}\) by marginalizing out \(S\) and
\[\begin{split}
    A^{-S}:=\{&p_{10.1},\\
              -&p_{00.0}-p_{01.0}+p_{10.1}+p_{01.1},\\
               &p_{10.0}+p_{01.0}-p_{00.1}-p_{01.1},\\
               &p_{10.0}\}
\end{split}\]
then \(L(o(\DF{\boldsymbol{V}\setminus S})):=\max(A^{-S})\) is a sharp lower bound of \(\theta(\DF{\boldsymbol{V}})\) under $G$ in \citep{Balke2011NonparametricBO}.

Let for each \(s\in\{0,1\}\)
\[\begin{split}
    A^{|s}:=\{&p_{10.1s},\\
          -&p_{00.0s}-p_{01.0s}+p_{10.1s}+p_{01.1s},\\
           &p_{10.0}+p_{01.0}-p_{00.1}-p_{01.1s},\\
           &p_{10.0s}\}
\end{split}\]
so \(L(o(\DF{\boldsymbol{V}|s}))=\max(A^{|s})\) is a sharp lower bound of \(\theta_{s}=\mathbb{E}_{P}[Y(X{=}0)|S{=}s]\) and let \(\bar{L}({o(\DF{\boldsymbol{V}})}):=\sum_{s\in\{0,1\}}L(o(\DF{\boldsymbol{V}|s}))p_s\).

Since \(S\indep Z\), we have for each \(x,y,z\in\{0,1\}\), \(p_{yx.z}=\sum_{s\in\{0,1\}}p_{yxs.z}=\sum_{s\in\{0,1\}}p_{yx.zs}p_{s.z}=\sum_{s\in\{0,1\}}p_{yx.zs}p_{s}=\mathbb{E}_P[p_{yx.zS}]\). Thus \(A^{-S}=\mathbb{E}_P[A^{|S}]\).

The covariate-optimal lower bound of \(\theta(\DF{\boldsymbol{V}})\), computed using \texttt{causaloptim} is given by \(L^{co}(\DF{\boldsymbol{O}})=\max(A)\) where
\[\begin{split}
    A:=\{&p_{100.1}+p_{101.1},\\
        -&p_{000.0}-p_{010.0}+p_{100.1}+p_{010.1}+p_{101.1},\\
         &p_{100.0}+p_{010.0}-p_{000.1}-p_{010.1}+p_{101.1},\\
         &p_{100.0}+p_{101.1},\\
         &p_{100.1}-p_{001.0}-p_{011.0}+p_{101.1}+p_{011.1},\\
        -&p_{000.0}-p_{010.0}+p_{100.1}+p_{010.1}-p_{001.0}-p_{011.0}+p_{101.1}+p_{011.1},\\
         &p_{100.0}+p_{010.0}-p_{000.1}-p_{010.1}-p_{001.0}-p_{011.0}+p_{101.1}+p_{011.1},\\
         &p_{100.0}-p_{001.0}-p_{011.0}+p_{101.1}+p_{011.1},\\
         &p_{100.1}+p_{101.0}+p_{011.0}-p_{001.1}-p_{011.1},\\
        -&p_{000.0}-p_{010.0}+p_{100.1}+p_{010.1}+p_{101.0}+p_{011.0}-p_{001.1}-p_{011.1},\\
         &p_{100.0}+p_{010.0}-p_{000.1}-p_{010.1}+p_{101.0}+p_{011.0}-p_{001.1}-p_{011.1},\\
         &p_{100.0}+p_{101.0}+p_{011.0}-p_{001.1}-p_{011.1},\\
         &p_{100.1}+p_{101.0},\\
        -&p_{000.0}-p_{010.0}+p_{100.1}+p_{010.1}+p_{101.0},\\
         &p_{100.0}+p_{010.0}-p_{000.1}-p_{010.1}+p_{101.0},\\
         &p_{100.0}+p_{101.0}\}
\end{split}\]
We show in the Supplementary Materials algebraically that \(\bar{L}(\DF{\boldsymbol{O}})=L^{co}(\DF{\boldsymbol{O}})\) and hence the covariate-averaged bounds in this setting are sharp. 

General conditions for sharpness under a given $G$ are more challenging to determine in this manner. However, it is easy to demonstrate that Assumptions 1-2 plus consistency are insufficient for covariate conditioning or covariate-averaging to yield sharp bounds for $\theta_s$ and $\theta$ respectively for all $G \in \mathcal{G}$, under a specific set of assumptions which include those encoded by the connectivity of $S$ in $G$. We now consider two examples where the covariate-averaged bounds are not sharp.

\subsection{Examples}
\subsubsection{Additional IV} \label{sec:addiv}

Consider the setting in Figure \ref{fig:simset3}, where $S, Z, X, Y$ are all binary. In this setting, our Assumptions 1-2 above hold, and thus Assumptions 3 and 4 of \citet{levis2023covariate} hold. In particular, their Assumption 3, which is $Z \indep (X(z), Y(z)) \vert S$, clearly holds, as does Assumption 4, which is $Y(x,z) = Y(x) \vert S$. %This is shown in the Supplementary Materials. 

For simplicity, we will only consider the first term of the risk difference: $\theta = P[Y(X = 1) = 1]$. Let $p_{yx\cdot zs} = P(Y = y, X = x \vert Z = z, S = s)$. The average Balke-Pearl lower bound is 
\begin{align*}
\bar{L}(\DF{\boldsymbol{O}}) = P(S = 1) \max\{ \\
1 - p_{00\cdot 01} - p_{01\cdot 01} - p_{10\cdot 0 1}, \\
1 - p_{00\cdot 1 1} - p_{01\cdot 1 1} - p_{10\cdot 1 1}, \\
1 - p_{01\cdot 0 1} - p_{10\cdot 0 1} - p_{00\cdot 1 1} - p_{01\cdot 1 1}, \\
1 - p_{00\cdot 0 1} - p_{01\cdot 0 1} - p_{01\cdot 1 1} - p_{10\cdot 1 1}
  \} + \\
P(S = 0) \max\{
1 - p_{00\cdot 0 0} - p_{01\cdot 0 0} - p_{10\cdot 0 0}, \\
1 - p_{00\cdot 1 0} - p_{01\cdot 1 0} - p_{10\cdot 1 0}, \\
1 - p_{01\cdot 0 0} - p_{10\cdot 0 0} - p_{00\cdot 1 0} - p_{01\cdot 1 0}, \\
1 - p_{00\cdot 0 0} - p_{01\cdot 0 0} - p_{01\cdot 1 0} - p_{10\cdot 1 0} 
  \}.
\end{align*}

We used \texttt{causaloptim} to derive the following covariate-optimal sharp lower bound on $\theta$ under Figure \ref{fig:simset3}. This method of derivation of the bounds has been proven to be sharp if the setting satisfies certain conditions, which are satisfied in this setting \citep{WOS:000801008700001}. The sharp covariate-optimal lower bound is 
\begin{align*}
{L}^{co}(\DF{\boldsymbol{V}}) = \max\{
1 - p_{00\cdot 0 0} - p_{01\cdot 0 0} - p_{10\cdot 0 0}, \\
1 - p_{00\cdot 1 0} - p_{01\cdot 1 0} - p_{10\cdot 1 0}, \\
1 - p_{00\cdot 0 1} - p_{01\cdot 0 1} - p_{10\cdot 0 1}, \\
1 - p_{00\cdot 1 1} - p_{01\cdot 1 1} - p_{10\cdot 1 1}, \\
1 - p_{01\cdot 0 0} - p_{10\cdot 0 0} - p_{00\cdot 0 1} - p_{01\cdot 0 1}, \\
1 - p_{00\cdot 0 0} - p_{01\cdot 0 0} - p_{01\cdot 0 1} - p_{10\cdot 0 1}, \\
1 - p_{01\cdot 0 0} - p_{10\cdot 0 0} - p_{00\cdot 1 0} - p_{01\cdot 1 0}, \\
1 - p_{01\cdot 0 1} - p_{10\cdot 0 1} - p_{00\cdot 1 0} - p_{01\cdot 1 0}, \\
1 - p_{00\cdot 0 0} - p_{01\cdot 0 0} - p_{01\cdot 1 0} - p_{10\cdot 1 0}, \\
1 - p_{00\cdot 0 1} - p_{01\cdot 0 1} - p_{01\cdot 1 0} - p_{10\cdot 1 0}, \\
1 - p_{01\cdot 1 0} - p_{10\cdot 1 0} - p_{00\cdot 1 1} - p_{01\cdot 1 1}, \\
1 - p_{01\cdot 0 0} - p_{10\cdot 0 0} - p_{00\cdot 1 1} - p_{01\cdot 1 1}, \\
1 - p_{01\cdot 0 1} - p_{10\cdot 0 1} - p_{00\cdot 1 1} - p_{01\cdot 1 1}, \\
1 - p_{00\cdot 0 1} - p_{01\cdot 0 1} - p_{01\cdot 1 1} - p_{10\cdot 1 1}, \\
1 - p_{00\cdot 0 0} - p_{01\cdot 0 0} - p_{01\cdot 1 1} - p_{10\cdot 1 1}, \\
1 - p_{00\cdot 1 0} - p_{01\cdot 1 0} - p_{01\cdot 1 1} - p_{10\cdot 1 1} 
\}.
\end{align*}

Note that each of the 8 terms inside the two maximum expressions of $\bar{L}$ are also contained in the maximum expression of $L^{co}$. However, there are more terms in the maximum expression of $L^{co}$. Thus we have 
\[
\bar{L} = p_s \max\{A_1\} + (1 - p_s) \max\{A_2\}
\]
and 
\[
L^{co} = \max\{A_1, A_2, B\}
\]
where $B \cap A_1 \cap A_2 = \varnothing$ and $p_s \in [0, 1]$. For a given distribution $\DF{\boldsymbol{V}}$, let $\tau_{\DF{\boldsymbol{V}}}$ denote the active term in the bound expressions, that is, the term or terms in the list of expressions which equals the maximum. Now we consider the following cases: 

\begin{itemize}
    \item If $\tau_{\DF{\boldsymbol{V}}} \in B$ and $\tau_{\DF{\boldsymbol{V}}} \notin A_1, \tau_{\DF{\boldsymbol{V}}} \notin A_2$ then clearly $\L^{co}(o(\DF{\boldsymbol{V}})) > \bar{L}(o(\DF{\boldsymbol{V}}))$. 
    \item If $\tau_{\DF{\boldsymbol{V}}} \in A_1$ or $(\tau_{\DF{\boldsymbol{V}}} \in A_1$ and $\tau_{\DF{\boldsymbol{V}}} \in B)$ (i.e., at least 2 terms are equal) then $\L^{co}(o(\DF{\boldsymbol{V}})) = \bar{L}(o(\DF{\boldsymbol{V}}))$ if and only if $p_s = 1$. But the causal model does not constrain $P(S = 1)$ in any way, thus we can find another distribution $P^*$ such that all the conditional probabilities are equal, but where $p_s < 1$, in which case $\L^{co}(o(\DF{\boldsymbol{V}})) > \bar{L}(o(\DF{\boldsymbol{V}}))$. The other cases can be handled similarly. 
\end{itemize}

The above demonstrates that there exists a distribution $\DF{\boldsymbol{V}}$ on $(Z, S, X, Y, U)$ from the causal model such that $\bar{L}(o(\DF{\boldsymbol{V}})) < L^{co}(o(\DF{\boldsymbol{V}}))$. Thus, $\bar{L}(o(\DF{\boldsymbol{V}}))$ is not sharp under $G$. 

\subsubsection{Identification with $S$} 

Consider the setting in Figure \ref{fig:dagsleft2}, where $S, Z, X, Y$ are all binary. In this setting, our Assumptions 1-2 hold, and thus Assumptions 3 and 4 of \citet{levis2023covariate} hold. 

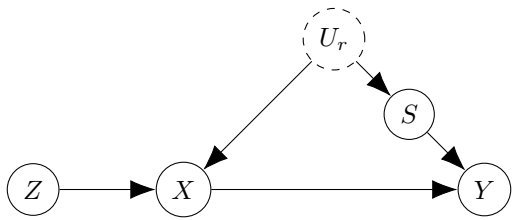
\begin{figure}[ht]
    \centering
    \begin{tikzpicture}
    \node[draw, circle, dashed] at (5, 2) (Ur) {\(U_r\)};
    \node[draw, circle] at (1, 0) (Z) {\(Z\)};
    \node[draw, circle] at (3, 0) (X) {\(X\)};
    \node[draw, circle] at (7, 0) (Y) {\(Y\)};
    \node[draw, circle] at (6, 1) (S) {\(S\)};
    \draw [-{Latex[scale=2.0]}] (Ur) to (X);
    \draw [-{Latex[scale=2.0]}] (Ur) to (S);
    \draw [-{Latex[scale=2.0]}] (Z) to (X);
    \draw [-{Latex[scale=2.0]}] (X) to (Y);
    \draw [-{Latex[scale=2.0]}] (S) to (Y);
   % \draw [-] (U) to (S);
    \end{tikzpicture}
    \caption{Causal IV setting with an unconfounded mediator \(S\)}
    \label{fig:dagsleft2}
\end{figure}

This example needs less explanation. Here, conditional on $S{=}s$, $\theta_s$ is point identified, and thus marginalizing over $S$, $\theta$ is also point identified via the back-door formula \cite{Pearl2000-PEAC-2}. The Balke-Pearl bounds conditioned on $S{=}s$ will not always collapse to $\theta_s$, and similarly averages of them will rarely, if ever, collapse to $\theta$. Thus, the Balke-Pearl bounds conditioned on levels on $S$ are not pointwise sharp, and averaging over them does not provide sharp bounds of $\theta$ in this setting.

We note that in \citet{levis2023covariate} Section 5.2, the authors state in our notation that one will want to covariate-average for $S$ whenever these covariates are predictive of $X$ and $Y$. However, our example here demonstrates that it is possible that covariates are predictive of $X$ and $Y$, but averaging over them can lead to bounds when instead considering the complete causal model can lead to point identification.

\section{Simulations}

To evaluate methods for bounding causal effects in the presence of unmeasured confounding but where at least one measured covariate is available, we simulated true probabilities from various models over \(S,Z,X,Y\) with all variables binary, and computed three different bounds based on them. The details of the simulation procedure are given in the Supplementary Materials. 

We consider variations of IV settings with a categorical covariate, expressed via causal diagrams shown in Figure \ref{fig:simsets}. Our estimand of interest is the causal risk difference $\theta = \mathbb{E}[Y(X = 1) - Y(X = 0)]$.
For each setting, we consider three methods of bounding the causal estimand: 
\begin{enumerate}
\item covariate-marginal (cm), which are the Balke-Pearl bounds,
\item covariate-averaged (ca), which are covariate-averaged bounds using the Balke-Pearl bounds conditional on $S$, and 
\item covariate-optimal (co) bounds derived using \texttt{causaloptim} including $S$. 
\end{enumerate}

%Since boun

\begin{figure}
    \centering
    \begin{subfigure}[b]{0.4\textwidth}
        \centering
        \begin{tikzpicture}[scale=0.8]
            \node[draw, circle, dashed] at (-4, 0) (Ul) {\(U_\ell\)};
            \node[draw, circle] at (-2, 2) (S) {\(S\)};
            \node[draw, circle] at (-2, 0) (Z) {\(Z\)};
            \node[draw, circle, dashed] at (1, 2) (Ur) {\(U_r\)};
            \node[draw, circle] at (0, 0) (X) {\(X\)};
            \node[draw, circle] at (2, 0) (Y) {\(Y\)};
            \draw [-{Latex[scale=2.0]}] (Ul) to (S);
            \draw [-{Latex[scale=2.0]}] (Ul) to (Z);
            \draw [-{Latex[scale=2.0]}] (S) to (Z);
            \draw [-{Latex[scale=2.0]}] (Ur) to (X);
            \draw [-{Latex[scale=2.0]}] (Ur) to (Y);
            \draw [-{Latex[scale=2.0]}] (Z) to (X);
            \draw [-{Latex[scale=2.0]}] (X) to (Y);
        \end{tikzpicture}
        \caption{Simulation setting (a)}
        \label{fig:simset1}
    \end{subfigure}
    \hfill
    \begin{subfigure}[b]{0.4\textwidth}
        \centering
        \begin{tikzpicture}[scale=0.8]
              \node[draw, circle, dashed] at (-4, 0) (Ul) {\(U_\ell\)};
            \node[draw, circle] at (-2, 2) (S) {\(S\)};
            \node[draw, circle] at (-2, 0) (Z) {\(Z\)};
            \node[draw, circle, dashed] at (1, 2) (Ur) {\(U_r\)};
            \node[draw, circle] at (0, 0) (X) {\(X\)};
            \node[draw, circle] at (2, 0) (Y) {\(Y\)};
            \draw [-{Latex[scale=2.0]}] (Ul) to (S);
            \draw [-{Latex[scale=2.0]}] (Ul) to (Z);
            \draw [-{Latex[scale=2.0]}] (Z) to (S);
            \draw [-{Latex[scale=2.0]}] (Ur) to (X);
            \draw [-{Latex[scale=2.0]}] (Ur) to (Y);
            \draw [-{Latex[scale=2.0]}] (Z) to (X);
            \draw [-{Latex[scale=2.0]}] (X) to (Y);
        \end{tikzpicture}
        \caption{Simulation setting (b)}
        \label{fig:simset2}
    \end{subfigure}
    \hfill
    \begin{subfigure}[b]{0.4\textwidth}
        \centering
        \begin{tikzpicture}[scale=0.8]
               \node[draw, circle, dashed] at (-4, 0) (Ul) {\(U_\ell\)};
            \node[draw, circle] at (-2, 2) (S) {\(S\)};
            \node[draw, circle] at (-2, 0) (Z) {\(Z\)};
            \node[draw, circle, dashed] at (1, 2) (Ur) {\(U_r\)};
            \node[draw, circle] at (0, 0) (X) {\(X\)};
            \node[draw, circle] at (2, 0) (Y) {\(Y\)};
            \draw [-{Latex[scale=2.0]}] (Ul) to (S);
            \draw [-{Latex[scale=2.0]}] (Ul) to (Z);
            \draw [-{Latex[scale=2.0]}] (S) to (Z);
            \draw [-{Latex[scale=2.0]}] (S) to (X);
            \draw [-{Latex[scale=2.0]}] (Ur) to (X);
            \draw [-{Latex[scale=2.0]}] (Ur) to (Y);
            \draw [-{Latex[scale=2.0]}] (Z) to (X);
            \draw [-{Latex[scale=2.0]}] (X) to (Y);
        \end{tikzpicture}
        \caption{Simulation setting (c)}
        \label{fig:simset3}
    \end{subfigure}
    \hfill
    \begin{subfigure}[b]{0.4\textwidth}
        \centering
        \begin{tikzpicture}[scale=0.8]
              \node[draw, circle, dashed] at (-4, 0) (Ul) {\(U_\ell\)};
            \node[draw, circle] at (-2, 2) (S) {\(S\)};
            \node[draw, circle] at (-2, 0) (Z) {\(Z\)};
            \node[draw, circle, dashed] at (1, 2) (Ur) {\(U_r\)};
            \node[draw, circle] at (0, 0) (X) {\(X\)};
            \node[draw, circle] at (2, 0) (Y) {\(Y\)};
            \draw [-{Latex[scale=2.0]}] (Ul) to (S);
            \draw [-{Latex[scale=2.0]}] (Ul) to (Z);
            \draw [-{Latex[scale=2.0]}] (Z) to (S);
            \draw [-{Latex[scale=2.0]}] (S) to (X);
            \draw [-{Latex[scale=2.0]}] (Ur) to (X);
            \draw [-{Latex[scale=2.0]}] (Ur) to (Y);
            \draw [-{Latex[scale=2.0]}] (Z) to (X);
            \draw [-{Latex[scale=2.0]}] (X) to (Y);
        \end{tikzpicture}
        \caption{Simulation setting (d)}
        \label{fig:simset4}
    \end{subfigure}
    \hfill
    \begin{subfigure}[b]{0.4\textwidth}
        \centering
        \begin{tikzpicture}[scale=0.8]
            \node[draw, circle] at (0, 2) (S) {\(S\)};
            \node[draw, circle] at (-2, 0) (Z) {\(Z\)};
            \node[draw, circle, dashed] at (2, 2) (Ur) {\(U_r\)};
            \node[draw, circle] at (0, 0) (X) {\(X\)};
            \node[draw, circle] at (2, 0) (Y) {\(Y\)};
            \draw [-{Latex[scale=2.0]}] (Ur) to (S);
            \draw [-{Latex[scale=2.0]}] (Ur) to (X);
            \draw [-{Latex[scale=2.0]}] (Ur) to (Y);
            \draw [-{Latex[scale=2.0]}] (S) to (X);
            \draw [-{Latex[scale=2.0]}] (S) to (Y);
            \draw [-{Latex[scale=2.0]}] (Z) to (X);
            \draw [-{Latex[scale=2.0]}] (X) to (Y);
        \end{tikzpicture}
        \caption{Simulation setting (e)}
        \label{fig:simset5}
    \end{subfigure}
    \hfill
    \begin{subfigure}[b]{0.4\textwidth}
        \centering
        \begin{tikzpicture}[scale=0.8]
                 \node[draw, circle, dashed] at (-4, 0) (Ul) {\(U_\ell\)};
            \node[draw, circle] at (-2, 2) (S) {\(S\)};
            \node[draw, circle] at (-2, 0) (Z) {\(Z\)};
            \node[draw, circle, dashed] at (1, 2) (Ur) {\(U_r\)};
            \node[draw, circle] at (0, 0) (X) {\(X\)};
            \node[draw, circle] at (2, 0) (Y) {\(Y\)};
            \draw [-{Latex[scale=2.0]}] (Ul) to (S);
            \draw [-{Latex[scale=2.0]}] (Ul) to (Z);
            \draw [-{Latex[scale=2.0]}] (Z) to (S);
            \draw [-{Latex[scale=2.0]}] (Ur) to (X);
            \draw [-{Latex[scale=2.0]}] (Ur) to (Y);
            \draw [-{Latex[scale=2.0]}] (S) to (X);
            \draw [-{Latex[scale=2.0]}] (X) to (Y);
        \end{tikzpicture}
        \caption{Simulation setting (f)}
        \label{fig:simset6}
    \end{subfigure}
    \caption{Simulation settings}
    \label{fig:simsets}
\end{figure}
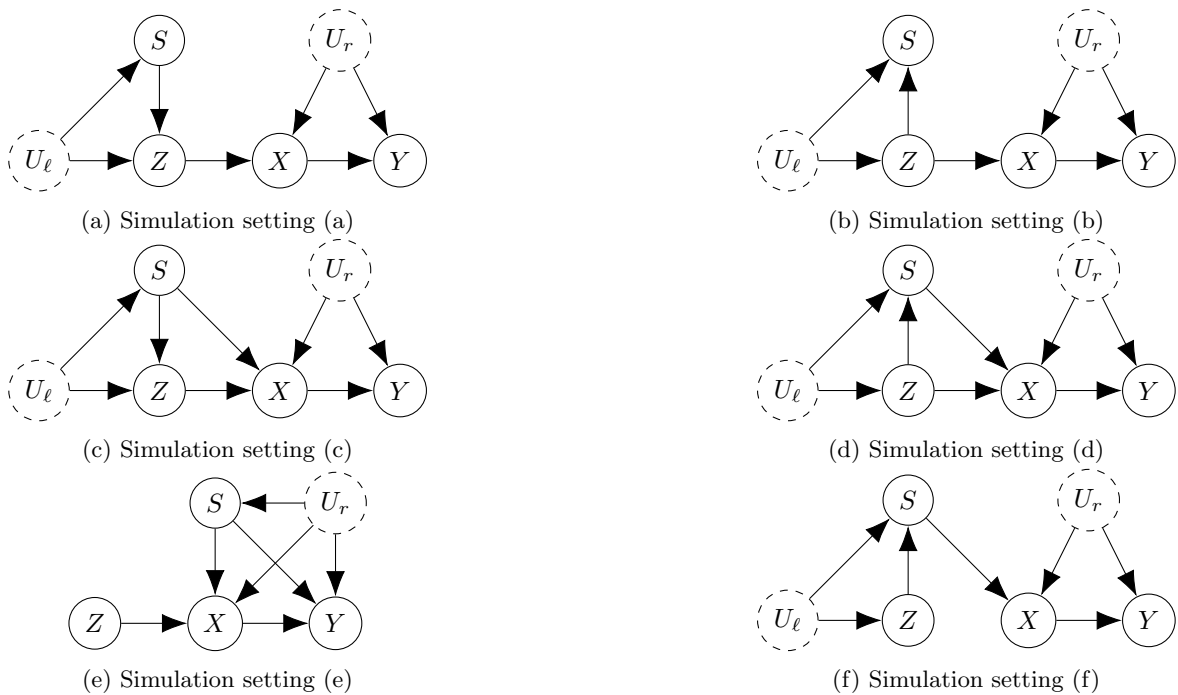

We performed \(n_{\text{sim}}:=10^5\) simulations for each of the following settings:
\renewcommand{\labelenumi}{\alph{enumi})}
\begin{enumerate}
    \item \(S\) as a confounded parent of \(Z\) as shown in Figure \ref{fig:simset1}
    \item \(S\) as a confounded child of \(Z\) as shown in Figure \ref{fig:simset2}
    \item \(S\) as a confounded parent of \(Z\) and a parent of \(X\) as shown in Figure \ref{fig:simset3}
    \item \(S\) as a confounded child of \(Z\) and a parent of \(X\) as shown in Figure \ref{fig:simset4}
    \item \(S\) is independent of \(Z\) and confounded parent of \(X\) and \(Y\) as shown in Figure \ref{fig:simset5}
    \item The effect of \(Z\) on \(X\) is entirely mediated via \(S\) as shown in Figure \ref{fig:simset6}
\end{enumerate}

\subsection{Results}
Every bound in every simulation was valid, as expected by theory for settings a-e, but potentially surprising in setting f. The covariate-optimal bounds are always contained in the covariate-averaged bounds and of course also in the covariate-marginalized bounds.
However, the covariate-averaged bounds are not contained in the covariate-marginalized bounds in settings d and f.
In setting f they are consistently worse.

Table \ref{tab:simresults0} shows the proportion of replicates where the bounds are contained in each other, and the average of the relative widths of the bounds for each setting in Figure \ref{fig:simsets} and for each of the three types of bounds. 
Figure \ref{fig:relwidths} shows the distributions of the relative widths of the covariate-optimal, covariate-averaged and covariate-marginalized bounds.

\begin{table}[ht]
    \centering
    \begin{tabular}{rlrrrrrr}
    \hline
        metric & model.a & model.b & model.c & model.d & model.e & model.f \\ 
        \hline
        co$\subseteq$cm & 1.00 & 1.00 & 1.00 & 1.00 & 1.00 & 1.00 \\ 
        co$\subseteq$ca & 1.00 & 1.00 & 1.00 & 1.00 & 1.00 & 1.00 \\ 
        ca$\subseteq$cm & 1.00 & 1.00 & 1.00 & 0.56 & 1.00 & 0.00 \\ 
        width(co)/width(cm) & 1.00 & 1.00 & 0.88 & 0.89 & 0.99 & 0.85 \\ 
        width(co)/width(ca) & 1.00 & 1.00 & 0.91 & 0.91 & 1.00 & 0.81 \\ 
        width(ca)/width(cm) & 1.00 & 1.00 & 0.97 & 0.98 & 0.99 & 1.05 \\ 
    \end{tabular}
    \caption{Results of simulations from models (a) to (f) for the covariate-marginalized bounds (cm), the covariate-averaged bounds (ca) and covariate-optimal bounds (co). \label{tab:simresults0}} 
\end{table}

\begin{figure}
    \centering
    \includegraphics[width=1.0\linewidth]{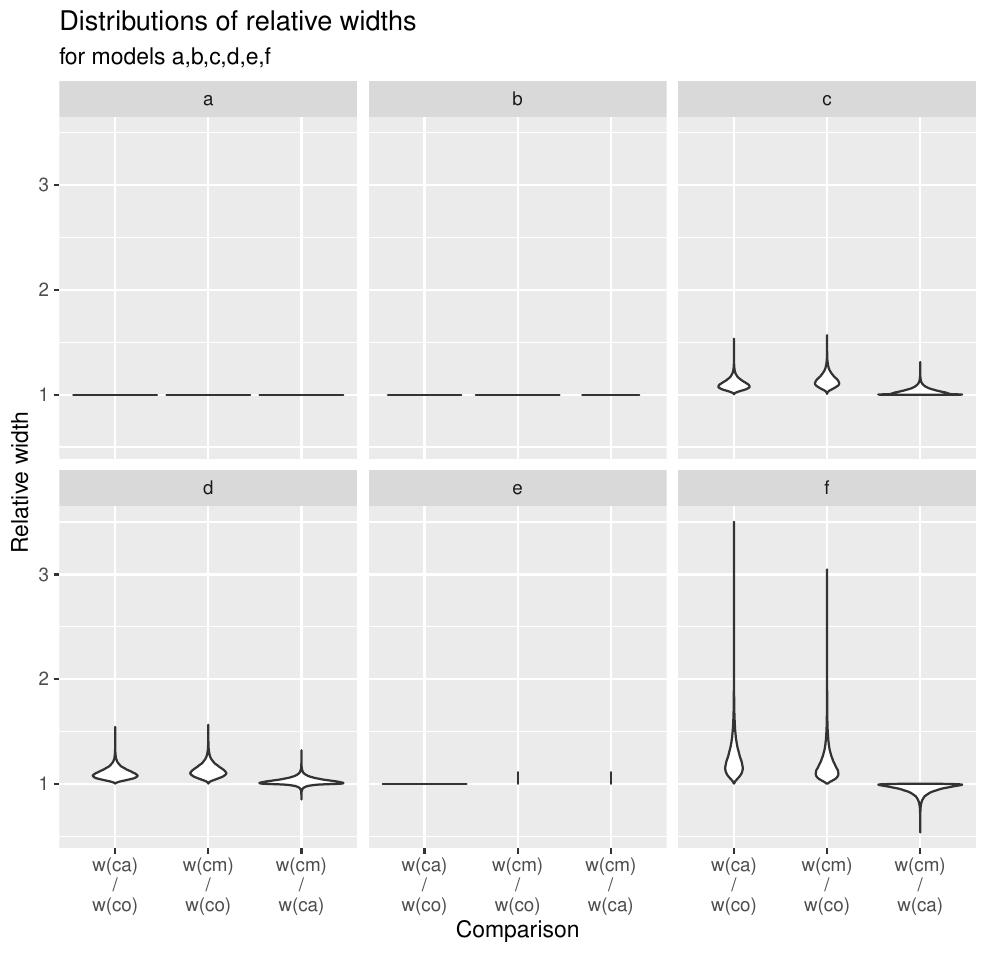}
    \caption{The distributions of width(ca)/width(co), width(cm)/width(co) and width(cm)/width(ca) for models a,b,c,d,e,f, where co, ca and cm are respectively the covariate-optimal, covariate-averaged and covariate-marginalized partial identification intervals.}
    \label{fig:relwidths}
\end{figure}

All three types of bounds coincide in Settings a and b and are also the same across the two settings. This is expected, since they both essentially reduce to a single binary IV setting.

In Setting c, where \(S\) and \(Z\) are confounded joint binary IVs for \(X\to Y\), which is equivalent to a 4-level IV, the covariate-optimal bounds are strictly narrower than the averaged or marginalized bounds.
The covariate-marginalized bounds are inferior to both because it has access only to the binary IV \(Z\) while the others use the IV \((S,Z)\) which provides more fine grained control of \(X\). The same phenomenon can be seen in Setting d, where the connection between \(S\) and \(Z\) is reversed, but the combined interpretation is similar.
Note that the covariate-marginalized bounds are exactly the same in both of these settings.

Setting e is the most general setting considered in \citet{https://doi.org/10.1002/sim.2766}, who showed the inclusion \text{ca}\(\subseteq\)\text{cm} and gave examples where this inclusion happens to be strict.
These examples relied on the of different strata to distinguish different compliance rates \(Z\to X\).
It is noteworthy that the covariate-averaged bounds coincide with the covariate-optimal bounds in this setting, whereas the covariate-marginalized bounds are strictly wider.

Setting f violates the condition of \citet{https://doi.org/10.1002/sim.2766}, that \(S\indep Z\), as well as that of \citet{levis2023covariate}, that conditioning on \(S\) renders \(Z\) a valid instrument (since relevance/association is broken).
The covariate-optimal bounds, however, are not affected and are furthermore strictly contained within the covariate-averaged bounds.

\section{Conclusions}
We have provided definitions of pointwise validity and sharpness that are in agreement with previous use in the literature but provide insight and clarity. We have additionally provided a definition for the novel concept of uniform sharpness, a concept used in \citet{levis2023covariate}, but not defined. 

We have proven two general propositions that can be used on a case-by-case basis to help determine in what settings covariate-conditional and covariate-averaged bounds will be valid and sharp. We have additionally proven a specific proposition that demonstrates that, in the all binary IV settings under specific conditions on a binary covariate, the set of covariate-conditional Balke-Pearl IV bounds are both pointwise and unformly sharp, and thus the covariate-averaged bounds are sharp for the marginal risk difference in this same setting. 

In several different IV settings, we demonstrate with simulations and examples where the covariate-conditional bounds and the covariate-averaged bounds are and are not pointwise and uniformly sharp for a given causal model $G$, and where covariate averaging can improve the width of the bounds, even if the bounds are not known to be sharp. We hope this demonstrates that sharpness is a nuanced concept requiring a clear statement of the assumed model or class of models, the estimand of interest, and the bounds under consideration.

We have demonstrated via examples that although Theorem 1 of \citet{levis2023covariate} is true as stated, it may or may not hold when further non-trivial information about $S$ is included in the causal model. General conditions, specifically graphical conditions, under which conditioning on $S$ will result in sharp bounds with respect to a specific casual model $G$ and causal estimand, are much more difficult to determine and is an area of future research for the authors.

\textbf{Funding information}: 
GJ, MCS, and EEG were partially supported by a grant from Novo Nordisk fonden NNF22OC0076595 and EEG and MCS the Pioneer Center for SMARTbiomed. 

\noindent
\textbf{Author contributions}: The authors jointly conceived the project. AS suggested the concept of uniform sharpness and derived its usefulness. GJ drafted the manuscript, wrote proofs and simulation code. EG and MCS contributed proofs and examples to the manuscript. All authors have commented on and contributed adjustments to the manuscript.

\noindent
\textbf{Conflict of interest}: No conflicts of interest to declare.

\noindent
\textbf{Data availability statement:} Code for reproducible simulations is available.

%In instrumental variable settings, where additional assumptions are not well-motivated, the average causal effect is not itself identifiable, but sharp bounds constraining the effect are identifiable. If additional covariates are available, then the instrumental variable conditions may hold within the strata of these covariates and marginalizing the within-strata bounds over the covariates may improve upon the sharp bounds in the setting without access to these covariates. This has been shown in a few settings, but generally may not provide the expected benefits of additional information that the covariates carry. In this article, we have shown examples where such averaging of bounds produce worse bounds than disregarding the covariates altogether. We caution to carefully consider the causal assumptions regarding additional covariates before applying this technique. In cases where optimal bounds can be computed on the causal assumptions that include the extra covariates, this will never be inferior and may sometimes be superior to averaging bounds conditional on the covariate.

\bibliographystyle{plainnat}
\bibliography{refs}

\end{document}